%% file: paper.tex
\documentclass[a4paper,10pt]{article}
\usepackage{amsmath,amssymb,amsthm,bbm,epsfig,graphics,graphicx,hyperref}
\usepackage{verbatim,braket,cleveref, enumerate}
\usepackage{color}
\usepackage{etoolbox}
\usepackage{authblk}
\newtheorem{theorem}{Theorem}
\newtheorem{lemma}[theorem]{Lemma}
\newtheorem{definition}[theorem]{Definition}
\newtheorem{remark}{Remark}
\newtheorem{example}[theorem]{Example}

\newtheorem{corollary}[theorem]{Corollary}

\newtheorem{subtheorem}{Theorem}[theorem]
\newbool{insubtheorem}
\AtBeginEnvironment{theorem}{\global\boolfalse{insubtheorem}}%
\AtBeginEnvironment{subtheorem}{%
  \ifbool{insubtheorem}{}{\global\booltrue{insubtheorem}%
    \refstepcounter{theorem}}%
}%

\newcommand{\HH}{\ensuremath{\mathcal{H}} }
\newcommand{\ZZ}{\ensuremath{\mathbb{Z}} }
\newcommand{\CC}{\ensuremath{\mathbb{C}} }
\newcommand{\Ccal}{\ensuremath{\mathcal{C}} }
\newcommand{\AME}[1]{AME($#1$)}
\newcommand{\QSS}[2]{$((#1,#2))$ threshold QSS scheme}
\newcommand{\floor}[1]{\ensuremath{\lfloor #1 \rfloor}}
\newcommand{\ceil}[1]{\ensuremath{\lceil #1 \rceil}}

\newcommand{\Id}{\mathbbm{1}}
\newcommand{\vv}[1]{\ensuremath{\mathbf{#1}}}

\DeclareMathOperator{\Tr}{Tr}

\crefname{equation}{Equation}{Equations}
\crefname{definition}{Defintion}{Defintions}
\crefname{item}{}{}

\def\noqed{\renewcommand{\qedsymbol}{}}

\title{Absolutely Maximally Entangled States: Existence and Applications}
\author{Wolfram Helwig}
\author{Wei Cui}
\affil{Center for Quantum Information and Quantum Control (CQIQC),\\
Department of Physics,
University of Toronto, Toronto, Ontario, M5S 1A7, Canada}

\begin{document}

\maketitle

\begin{abstract}
We investigate absolutely maximally entangled (AME) states, which are multipartite quantum states that are maximally entangled with respect to any possible bipartition. These strong entanglement properties make them a powerful resource for a variety of quantum information protocols. In this paper, we show the existence of AME states for any number of parties, given that the dimension of the involved systems is chosen appropriately. We prove the equivalence of AME states shared between an even number of parties and pure state threshold quantum secret sharing (QSS) schemes, and prove necessary and sufficient entanglement properties for a wider class of ramp QSS schemes. We further show how AME states can be used as a valuable resource for open-destination teleportation protocols and to what extend entanglement swapping generalizes to AME states.
\end{abstract}

\section{Introduction}
Entanglement has been a hot topic since the beginning of quantum mechanics and fueled a lot of discussions, among them most notable the Einstein-Podolsky-Rosen (EPR) paradox \cite{Einstein1935}, which finally led Bell to come up with a method of actually measuring entanglement \cite{Bell1964}. It was not until the advent of quantum information, however, that entanglement was recognized as a useful resource. Almost all applications in quantum information make either explicit or implicit use of entanglement, which makes it crucial to gain as much insight as possible. \cite{Horodecki2009}

While the entanglement of bipartite states is already very well understood \cite{Bennett1996a, Nielsen1999, Vidal1999}, the road to its generalization to more than two parties is paved with many obstacles. Therefore we often have to restrict ourselves to special cases when analyzing multipartite entanglement. A prominent choice are states that extremize the entanglement for a certain measure of entanglement. In this paper we want to do that by focusing on \emph{absolutely maximally entangled (AME)} states, which are defined as states that are maximally entangled for any possible bipartition. \cite{Helwig2012, Gisin1998, Gour2010}

\begin{definition}
\label{def:AME}
	An absolutely maximally entangled state is a pure state, shared among $n$ parties $P=\{1,\ldots,n\}$, each having a system of dimension $d$. Hence $\ket{\Phi} \in \HH_1 \otimes \cdots \otimes \HH_n$, where $\HH_i \cong \CC^d$, with the following equivalent properties:
	\begin{enumerate}[(i)]
		\item $\ket{\Phi}$ is maximally entangled for any possible bipartition. This means that for any bipartition of $P$ into disjoint sets $A$ and $B$ with $A\cup B  = P$ and, without loss of generality, $m=|B|\leq |A|=n-m$, the state $\ket{\Phi}$ can be written in the form 
		\begin{equation}
			\label{eq:defAMEstate}
			\ket\Phi = 
			\frac{1}{\sqrt{d^m}}\sum_{k\in \ZZ_d^{m}} 
			\ket{k_1}_{B_1}\cdots \ket{k_{m}}_{B_{m}}
			\ket{\phi(k)}_A,
		\end{equation}
		with $\braket{\phi(k)|\phi(k')} = \delta_{kk'}$.
		\item The reduced density matrix of every subset of parties $A\subset P$ with $|A| = \floor{\frac{n}{2}}$ is totally mixed, $\rho_A = d^{-\floor{\frac{n}{2}}} \Id_{d^{\floor{\frac{n}{2}}}}$.
		\item The reduced density matrix of every subset of parties $A\subset P$ with $|A|\leq \frac{n}{2}$ is totally mixed.
		\item The von Neumann entropy of every subset of parties $A\subset P$ with $|A| = \floor{\frac{n}{2}}$ is maximal, $S(A) = \floor{\frac{n}{2}} \log d$.
		\item The von Neumann entropy of every subset of parties $A\subset P$ with $|A| \leq \frac{n}{2}$ is maximal, $S(A) = |A| \log d$.
\label{def:AMEentropy}
	\end{enumerate}
	These are all necessary and sufficient condition for a state to be absolutely maximally entangled. We denote such a state as an \AME{n,d} state.
\end{definition}

The simplest examples of AME states occur for low dimensional systems shared among few parties. Starting with qubits, the most obvious one is an EPR pair, which is maximally entangled for its only possible bipartition. For three qubits shared among three parties, we can recognize the GHZ state as an AME state. It is maximally entangled, with 1 ebit of entanglement with respect to every bipartition. For four qubits, there is no obvious candidate, and in fact it has been shown that for four qubits no AME state exists \cite{Gour2010}. We can still find an absolutely maximally entangled states for four parties, however, by increasing the dimensions of the involved systems. An \AME{4,3} state for four \emph{qutrits} shared among four parties exists, and it is given by \cite{Helwig2012}
\begin{gather}
	\label{eq:popescuestate}
	\begin{split}
	\ket{\Phi} &= \frac{1}{\sqrt{9}} \sum_{i,j=0}^2 \ket{i} \ket{j} \ket{i+j} \ket{i+2j}.
	\end{split}
\end{gather}
This is the first indicator that the search for AME states gets more promising as we increase the dimensions of the systems.

Completing the characterization of AME states for qubits, it is known that AME states exist for 5 and 6 qubits. Explicit forms for them are given in Ref.~\cite{Helwig2012}, and it turns out that they are closely related to the five-qubit error correction code. For 7 qubits, it is still not known if an AME state exists, whereas for $\geq 8$ qubits, it has been shown that no AME states can exist \cite{Gour2010, Rains1999}.

In Ref.~\cite{Helwig2012}, we showed how AME states can be used for parallel teleportation protocols. In these protocols, the parties are divided into a sets of senders and receivers, respectively. One of the two sets is given the ability to perform joint quantum operations, while players in the other set can only perform local quantum operations. Under these conditions, a parallel teleportation of multiple quantum states is possible if the set that performs joint quantum operations is larger than the other set.
A closer look at these teleportation scenarios then led to the observation that any AME state shared by an even number of parties can be used to construct a threshold quantum secret sharing (QSS) scheme \cite{Cleve1999, Gottesman2000, Imai2005}. The opposite direction was also shown, with one additional condition imposed on the QSS scheme, namely that the shared state that encodes the secret is already an AME state.

In this paper, we will give an information-information theoretic proof of this equivalence of AME states and threshold QSS scheme, which shows that the additional condition is not required. We will rather see that it is satisfied for all threshold QSS schemes. We will further give a recipe of how to construct AME states from classical codes that satisfy the Singleton bound \cite{Singleton1964}. This construction can be used to 
produce AME states for a wide class of parameters, and it even proves that AME states exist for any number of parties for appropriate system dimension. A result that could also be deduced from the equivalence of AME states and QSS schemes and a known construction for threshold QSS schemes \cite{Cleve1999}. We will then show more applications for AME states. The first being the construction of a wider class of QSS schemes, the \emph{ramp} QSS schemes, of which threshold QSS schemes are a special case. The next one is the utilization of AME states as resources for open-destination teleportation protocols \cite{Zhao2004}. Finally, we investigate to what extend entanglement can be swapped between two AME states.

This paper is structured as follows. In Section~\ref{section:mds}, we show how AME states can be constructed from classical codes, which also also shows the existence of AME states for any number of parties. In Section~\ref{section:AME-QSS}, we establish an equivalence between even party AME states 
and threshold QSS schemes, using an information theoretic approach to QSS schemes.
Section~\ref{section:muitisecret} shows how to share multiple secrets using AME states. In Section~\ref{section:openteleport}, we show that AME states can be used for open-destination teleportation. After that, swapping of AME states is investigated in Section~\ref{section:swapping}.

\section{Constructing AME States from Classical MDS Codes}
\label{section:mds}
There is a subclass of AME($n$,$d$) states that can be constructed from optimal classical error correction codes.  A classical code \Ccal consists of $M$ codewords of length $n$ over an alphabet $\Sigma$ of size $d$. For our purposes, the alphabet is going to be $\Sigma=\ZZ_d$ and thus $\Ccal \subset \ZZ_d^n$. The \emph{Hamming distance} between two codewords is defined as the number of positions in which they differ, and the minimal distance $\delta$ of the code \Ccal as the minimal Hamming distance between any two codewords. For a given length $n$ and minimal distance $\delta$, the number of codewords $M$ in the code is bounded by the \emph{Singleton bound} \cite{Singleton1964, MacWilliams1977}
\begin{equation}
 \label{eq:Singleton}
 M \leq d^{n-\delta+1}.
\end{equation}
Codes that satisfy the Singleton bound are referred to as maximum-distance separable (MDS) codes. They can be used to construct AME states:

\begin{subtheorem}
  \label{theorem:MDS}
  From a classical MDS code $\Ccal \subset \ZZ_d^{2m}$ of length $2m$ and minimal distance $\delta = m+1$ over an alphabet $\ZZ_d$, an \AME{2m,d} state can be constructed as
  \begin{align}
    \label{eq:MDSstate}
    \ket{\textrm{AME}}
      &=\frac{1}{\sqrt{d^{m}}}\sum_{c\in \Ccal} \ket{c}\\
	&= \frac{1}{\sqrt{d^{m}}}\sum_{c\in \Ccal}
	\ket{c_1}_1 \cdots \ket{c_m}_m \ket{c_{m+1}}_{m+1} \cdots \ket{c_{2m}}_{2m}.
  \end{align}
\end{subtheorem}
\begin{proof}
  The code \Ccal satisfies the Singleton bound, which means the sum contains a total of $M=d^{2m-\delta+1}=d^m$ terms. Furthermore, any two of these terms differ in at least one of the first $m$ kets because the code has minimal distance $\delta = m+1$. Hence the sum contains each possible combination of the first $m$ basis kets exactly once. Moreover, for any two different terms, the last $m$ kets must also differ in at least one ket and are thus orthogonal. This means the state has the form of Equation~\eqref{eq:defAMEstate} with respect to the bipartition into the first $m$ and last $m$ parties. The same argument works for any other bipartition into two sets of size $m$, hence the state is absolutely maximally entangled.
\end{proof}

An analogous argument shows that a similar construction for an odd number of parties results in an AME state.
\begin{subtheorem}
 \label{theorem:MDSodd}
  From a classical MDS code $\Ccal \subset \ZZ_d^{2m+1}$ of length $2m+1$ and minimal distance $\delta = m+2$ over an alphabet $\ZZ_d$, an \AME{2m+1,d} state can be constructed as
  \begin{align}
    \label{eq:MDSstateodd}
    \ket{\textrm{AME}}
      &=\frac{1}{\sqrt{d^{m}}}\sum_{c\in \Ccal} \ket{c}\\
	&= \frac{1}{\sqrt{d^{m}}}\sum_{c\in \Ccal}
	\ket{c_1}_1 \cdots \ket{c_{m+1}}_{m+1} \ket{c_{m+2}}_{m+2} \cdots \ket{c_{2m}}_{2m+1}.
  \end{align}
\end{subtheorem}
\begin{proof}
 The code contains $M=d^{m}$ terms. Each of the terms differ in at least one of the first $m+1$ and last $m$ terms. Thus, with the same argument as above, this is an AME state.
\end{proof}

Trivial states of that form are $d$-dimensional EPR states, which are represented by the code with codewords $00, 11, \ldots, (d-1)(d-1)$. This code has $n=2$, $\delta = 2$, $M=d^1$. For $n=3$, we can find the GHZ states for arbitrary dimensions, which can be constructed from the code $000, 111, \ldots, (d-1)(d-1)(d-1)$, which has $\delta = 3$ and $M=d^1$. As already mentioned in the introduction, for $n=4$ no AME state exists for $d=2$, however for $d=3$ the \AME{4,3} state given in Equation~\eqref{eq:popescuestate} can also be constructed from an MDS code, the $[4,2,3]_3$ ternary Hamming code.

A wide class of MDS codes is given by the Reed-Solomon codes and its generalizations \cite{Reed1960, MacWilliams1977, Seroussi1986}, which give MDS codes for $n=d-1$, $n=d$, and $n=d+1$, for $d=p^x$ being a positive power of a prime number $p$. From the Reed-Solomon codes, MDS codes can also be constructed for $n<d-1$ \cite{Singleton1964}. This shows that AME states exist for any number of parties if the system dimensions are chosen right.

At this point we would like to mention that after posting a preliminary version of our last paper on this subject \cite{Helwig2012}, it has been brought to our attention by Gerardo Adesso that the results of this section have already been previously discovered by Ashish Thapliyal and coworkers, and were presented at a conference in 2003 \cite{Thapliyal2003}, but remained unpublished.

\section{Equivalence of AME states and QSS schemes}
\label{section:AME-QSS}
In Ref.~\cite{Helwig2012}, we showed that \AME{2m,d} states, i.e., AME states shared between an even number of parties, are equivalent to pure state threshold quantum secret sharing (QSS) schemes that have AME states as basis states and share and secret dimension equal to $d$. Here we will give an information-theoretic proof of this equivalence, which shows that the requirement that the basis states of the QSS scheme are AME states is redundant, as it follows from this proof that these states are always absolutely maximally entangled. Before stating the theorem and the proof, we give a short motivation why AME states and QSS schemes are related.

Consider an \AME{2m,d} state shared among an even number of parties. If we take any bipartition into two sets of parties $A$ and $B$, each of size $m$, a $d^m$ dimensional state can be teleported from one set to the other due to the maximal entanglement between $A$ and $B$. Moreover, we have shown in Ref.~\cite{Helwig2012}, that the teleportation can be performed in such a way that each party in the sending set $B$ performs a local teleportation operation on their qudit, while the parties in the receiving set $A$ perform a joint quantum operation to recover all $m$ teleported qudits. This is depicted in Figure~\ref{fig:teleport} for the case of $m=4$. This also works if only one party in $B$, which we call the \emph{dealer} $D$, performs the teleportation operation, while the others do nothing. Then the teleported $d$-dimensional state can still be recovered by the players in set $A$. Furthermore, this also works for any other bipartition into sets $A'$ and $B'$ of size $m$, with $D\in B'$, without changing the teleportation operation $D$ has to 
perform, 
but now the parties in $A'$ can recover the teleported state (see Figure~\ref{fig:ameqss}). This means that any set with $m$ parties can recover the state. Moreover, the no-cloning theorem guarantees that the complement of a set that can recover the state has no information about the state. Hence all sets with less than $m$ parties cannot gain any information about the state. This, however, are exactly the requirements for a threshold QSS scheme, therefore we have constructed a \QSS{m}{2m-1} from the \AME{2m,d} state. To formally show this, and moreover that it also works in the opposite direction, meaning that a \QSS{m}{2m-1} is always related to an \AME{2m,d} state, we will use  the information theoretic description of QSS schemes as introduced in Ref.~\cite{Imai2005}.

\begin{figure}[htb]
	\centering
	\includegraphics[width=.3\textwidth]{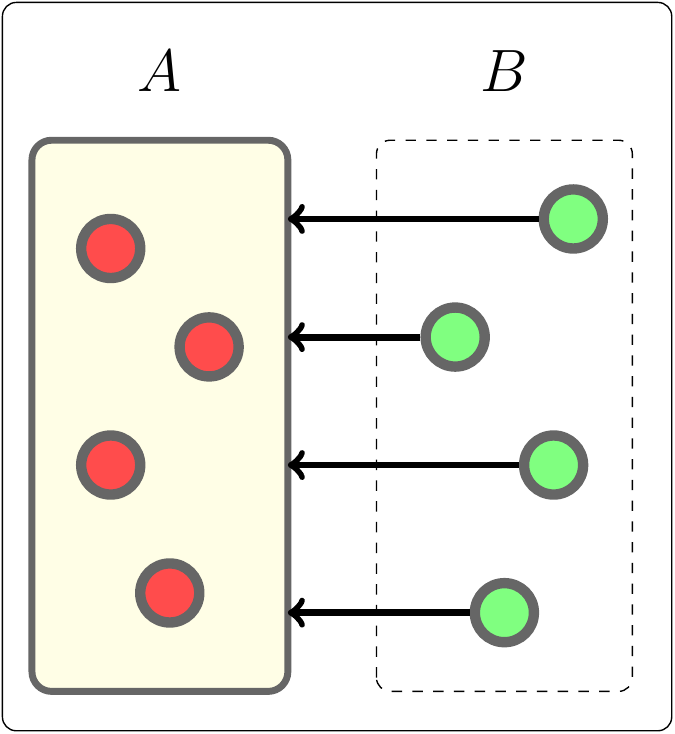}
	\caption{(Color online) Parties in $B$ (green) perform local teleportation operations, parties in $A$ (red) can recover teleported states by performing a joint quantum operation}
	\label{fig:teleport}
\end{figure}

\begin{figure}[bht]
	\centering
	\includegraphics[width=0.8\textwidth]{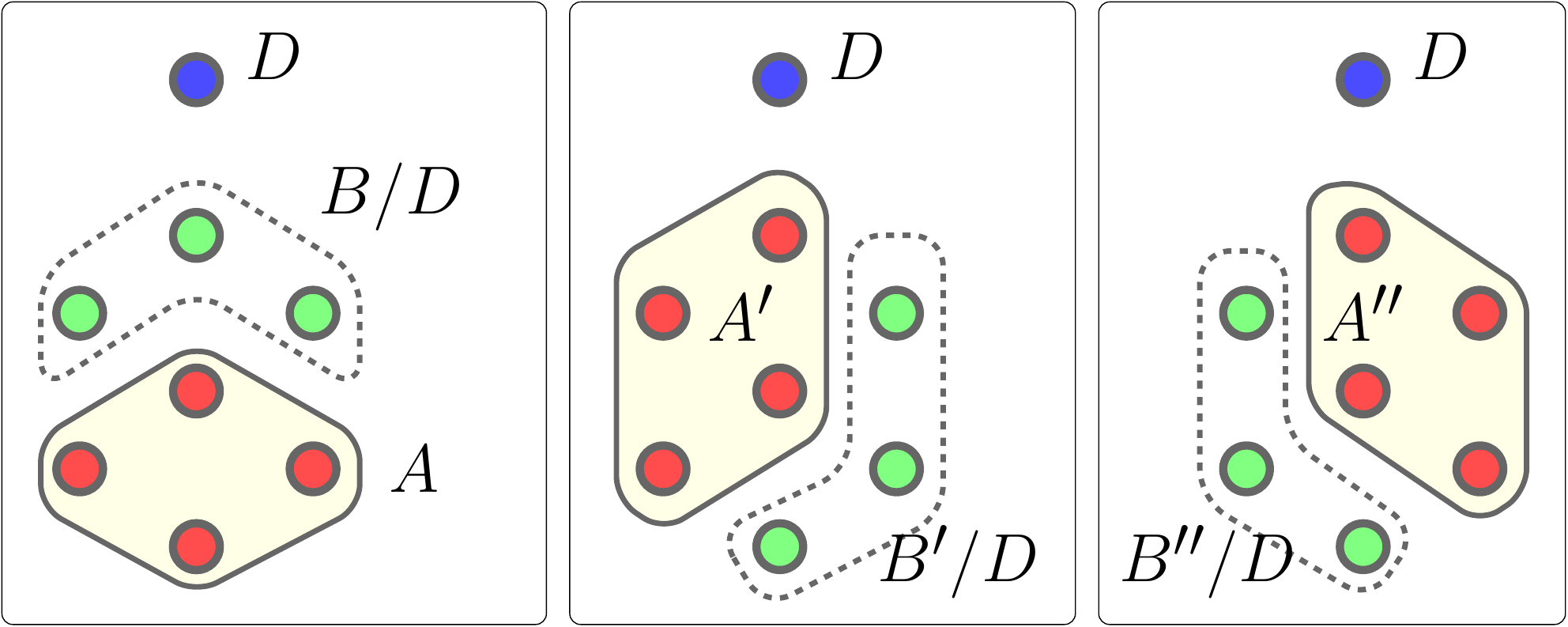}
	\caption{(Color online) After $D$ (blue) performs her teleportation operation, any set of $m$ parties (red), $A$, $A'$, $A''$ etc., can recover the teleported state. Any set of parties with $m-1$ or less parties (any set consisting only of green parties) cannot gain any information about the teleported state.}
	\label{fig:ameqss}
\end{figure}

Let us quickly review the framework for a pure state \QSS{m}{2m-1} \cite{Cleve1999}. A secret $S$ is distributed among the players $P=\{1,\ldots,2m-1\}$ such that any set $A\subseteq P$ with $|A| \geq m$ can recover the secret, while any set $B \subset P$ with $|B| < m$ cannot gain any information about the secret. We further only consider the case where the dimension $d$ of the secret is the same as the dimension of each player's share.  

The secret is assumed to lie in the Hilbert space $\HH_S \cong \CC^d$, and the share of party $i$ in $\HH_i \cong \CC^d$. The encoding is described by an isometry
\begin{equation}
	\label{eq:encoding}
	U_S: \HH_S \rightarrow \HH_1 \otimes \dots \otimes \HH_{2m-1}.
\end{equation}
The secret $S$ is chosen randomly and thus is described by $\rho_S = 1/d \sum_i \ket{i}\bra{i}$. We consider its purification by introducing a reference system $R$ such that $\ket{RS} = 1/\sqrt{d} \sum_i \ket{i}\ket{i} \in \HH_R \otimes \HH_S$. Let $\rho_{RA}$ denote the combined state of the reference system and a set of players $A\subseteq P$ after $U_S$ has been applied to the secret. Then the players $A$ can recover the secret, if there exists a completely positive map $T_A: \HH_A \rightarrow \HH_S$ such that \cite{Imai2005, Schumacher1996}
\begin{equation}
	\label{eq:RA_recover}
	\Id_R \otimes T_A (\rho_{RA}) = \ket{RS}.
\end{equation}
This can be stated in terms of the mutual information
\begin{equation}
	\label{eq:mutualinfo}
	I(X:Y) = S(X) + S(Y) - S(X,Y)
\end{equation}
as follows:
\begin{definition}
	An isometry $U_S: \HH_S \rightarrow \HH_1 \otimes \dots \otimes \HH_{2m-1}$ creates a \QSS{m}{2m-1} if and only if, after applying to the system $S$ of the purification $\ket{RS}$, the mutual information between $R$ and an authorized (unauthorized) set of players $A$ ($B$) satisfies
\begin{align}
	\label{eq:QSSconditionA}
	I(R:A) &= I(R:S) = 2S(S) && \text{if } |A| \geq m\\
	\label{eq:QSSconditionB}
	I(R:B) &= 0 && \text{if } |B| < m.
\end{align}
\end{definition}
Here $S$ is the von Neumann entropy, and because of $S(i)\geq S(S)$ \cite{Imai2005}, we have
\begin{equation}
	\label{eq:secretentropy}
	S(S) = S(R) = S(i) = \log d.
\end{equation}
From \cref{eq:mutualinfo,eq:QSSconditionA,eq:QSSconditionB} it immediately follows that
\begin{align}
	\label{eq:entropyRA}
	S(R,A) &= S(A) - S(R) \quad \text{if } |A| \geq m\\
	\label{eq:entropyRB}
	S(R,B) &= S(B) + S(R) \quad \text{if } |B| < m.
\end{align}

\begin{theorem}
	\label{theorem:AME-QSS}
For a state $\ket{\Phi}$ the following two properties are equivalent:
\begin{enumerate}[(i)]
	\item $\ket{\Phi}$ is an \AME{2m,d} state.
	\item $\ket{\Phi}$ is the purification of a \QSS{m}{2m-1}, whose share and secret dimensions are $d$.
\end{enumerate}
\end{theorem}
\begin{proof}
$(i) \rightarrow (ii)$:
We need to show that for an \AME{2m,d} state \cref{eq:QSSconditionA,eq:QSSconditionB} are satisfied, where $R$ can be any of the $2m$ party. This follows directly from the definition of the mutual information, \cref{eq:mutualinfo}, and \cref{def:AME} (\ref{def:AMEentropy}).

$(ii)\rightarrow (i)$:
Consider an unauthorized set of players $B$, with $|B| = m-1$. Then the set is $B \cup i$ is authorized for any additional player $i \notin B$, and from Equation~\eqref{eq:entropyRA} we have
\begin{equation}
 S(B,i,R) = S(B,i) - S(R)
\end{equation}
On the other hand, using the Araki-Lieb inequality \cite{Nielsen2000} $S(X,Y) \geq S(X) - S(Y)$ and Equation~\eqref{eq:entropyRB} gives
\begin{equation}
 S(B,i,R) \geq S(B,R) - S(i) = S(B) + S(R) - S(i).
\end{equation}
Combining the last two equations and using $S(S)=S(R)=S(i)$ shows
\begin{equation}
 S(B,i) \geq S(B) + S(i),
\end{equation}
where equallity must hold due to the subadditivity of the entropy $S(X,Y) \leq S(X) + S(Y)$. This means that the entropy increases maximally when adding one player's share to $m-1$ shares. The strong subadditivity of the entropy \cite{Nielsen2000}
\begin{equation}
 S(X,Y) - S(Y) \geq S(X,Y,Z) - S(Y,Z)
\end{equation}
states that adding one system $X$ to a system $Y$ increases the entropy at least by as much as adding the system $X$ to a larger system $Y \cup Z$ that contains $Y$. So in our case, adding one share to less than $m-1$ shares increases the entropy by at least $S(i)$, and since this is the maximum, it increases the entropy exactly by $S(i)$. Hence, starting out with a set of no shares, and repeatedly adding one share to the set until the set contains any $m$ shares and is authorized, shows that any set of $m$ shares has entropy $m S(i)$. This shows that the entropy is maximal for any subset of $m$ parties and thus $\ket{\Phi}$ is an \AME{2m,d} state.
\end{proof}

\begin{corollary}
The encoded state $U_S \ket{S}$ of a specific secret $\ket{S}$ with a $((m,2m-1))$ threshold QSS protocol with share and secret dimension $d$ is an \AME{2m-1,d} state.
\end{corollary}

\section{Sharing multiple secrets}
\label{section:muitisecret}
In the previous section, we outlined how an AME state can be used to construct a QSS scheme. The role of the dealer is assigned to one of the parties and he performs a teleportation operation on his qudit, which encodes the teleported qudit onto the qudits of the remaining parties such that the criteria for a QSS scheme are met. While Theorem~\ref{theorem:AME-QSS} shows the equivalence of AME states and QSS schemes, the actual protocol for the encoding and decoding operations has been presented in Ref.~\cite{Helwig2012}. Note that in the described scenario, the role of the dealer can be assigned to any player. Thus one may ask, what happens if more than one of the players assumes the role of the dealer. The answer is that, given an \AME{2m, d} state, up to $m$ players are able to independently encode one qudit each onto the qudits of the remaining players in such a way that results in a QSS scheme with a more general access structure.

For a secret sharing scheme with a general access structure, each set of players falls into one of three categories \cite{Iwamoto2005, Gheorghiu2012}.
\begin{enumerate}
 \item \emph{Authorized}: A set of players is authorized, if it can recover the secret
 \item \emph{Forbidden}: A set of players is called a forbidden set, if the players cannot gain any information about the encoded secret
 \item \emph{Intermediate}: A set of players is classified as an intermediate set, if they cannot recover set secret, but may be able to gain part of the information. This means that the reduced density matrix of that set of players depends on the encoded secret, but not enough as to recover the secret.
\end{enumerate}

A special kind of access structure is a $(m,L,n)$ \emph{ramp secret sharing scheme} \cite{Blakley1984}. Here $n$ is the total number of players, $m$ is the number of players needed to recover the secret, and $L$ is the number of shares that have to be removed from a minimal authorized set to destroy all information about the secret. In terms of the above defined set categories that means that any set of $m$ or more players is authorized,  any set of $m-L$ or less players is forbidden, and any set consisting of more than $m-L$, but less than $m$ players is an intermediate set. This is the access structure we get from an \AME{2m,d} state if more than one party assumes the role of the dealer.

\begin{theorem}
 \label{theorem:multisecrets}
  Given an \AME{2m,d} state, a QSS scheme with secret dimension $d^L$ and a $(m,L,2m-L)$ ramp access structure can be constructed for all $1\leq L \leq m$.
\end{theorem}
\begin{proof}
 The encoding of the secret is done by assigning the role of dealer to $L$ of the $2m$ players. For simplicity we choose them to be the first $L$ players. Each of them performs a Bell measurement on their respective qudit of the AME state and one qudit of the secret. The Bell measurement is described by the general $d$-dim Bell states $\ket{\Psi_{kl}}$ and the unitaries $U_{kl}$ that transform among them \cite{Bennett1993}
\begin{align}
	\label{eq:Bell}
	\ket{\Psi_{qp}} &= \frac{1}{\sqrt{d}}\sum_j e^{2\pi ijq/d} \ket{j}\ket{j+p}\\
	\label{eq:BellU}
	U_{qp} &= \sum_j e^{2\pi ijq/d} \ket{j}\bra{j+p},
\end{align}
where the kets are understood to be mod $d$. For a secret $\ket{s}$ and outcomes $(q_1,p_1) \ldots (q_L,p_L)$ for the Bell measurement of the dealers, the initial \AME{2m,d} state is transformed to
\begin{equation}
  \label{eq:PhiS}
  \ket{\Phi_S} = \frac{1}{\sqrt{d^{m-L}}}
  \sum_{k\in \ZZ_d^m} s_{\vv{q}\vv{p},k_1\cdots k_L}
  \ket{k_{L+1}}_{B_1} \cdots \ket{k_m}_{B_{m-L}}
  \ket{\phi(k)}_{A}.
\end{equation}
Here
\begin{equation}
  \label{eq:s}
 s_{\vv{q}\vv{p},k_1\cdots k_L} = 
  \braket{k_1\cdots k_L|U_{q_1p_1}^{\dagger}\otimes \cdots \otimes U_{q_Lp_L}^{\dagger}|s},
\end{equation}
and the partition of the remaining $2m-L$ parties into two sets $A$ and $B$ of size $m$ and $m-L$, respectively, is arbitrary. After obtaining their measurement outcomes, the dealers broadcast their results to all of the remaining players. This concludes the encoding process.

To show that any set of $m$ or more players is authorized, it suffices to show that set $A$ in Equation~\ref{eq:PhiS} can recover the secret. They can do so by applying the unitary operation
\begin{align}
  U = ( U_{q_1p_1} \otimes \cdots \otimes U_{q_Lp_L} \otimes \Id ) V
  \intertext{with}
  \label{eq:sortV}
  V = \sum_{k\in \ZZ_d^m} \ket{k_1}\cdots\ket{k_m} \bra{\phi(k)},
\end{align}
to their system. This changes the state to
\begin{equation}
 U \ket{\Phi_S} = \frac{1}{\sqrt{d^{m-L}}}
  \sum_{(k_{L+1},\ldots,k_m)\in \ZZ_d^{m-L}}
  \ket{k_{L+1}}_{B_1} \cdots \ket{k_m}_{B_{m-L}}
  \ket{s}_{A'}
  \ket{k_{L+1}}_{A_{L+1}} \cdots \ket{k_m}_{A_m}
\end{equation}
where $A' = \{ A_1, \ldots, A_L \}$. Thus the players in set $A$ have the secret in their possession. It immediately follows from the no-cloning theorem that $B$, and thus any set of size $m-L$ or less, cannot have any information about the secret since all information is located in the complement set. Alternatively, this also follows from the observation that the reduced density matrix of $B$ is always completely mixed, independent of the secret.

The last thing left to show is that all sets with more than $m-L$ but fewer than $m$ players are indeed intermediate sets. To see that, consider the case $L=1$, where a set $C$ of $m-1$ players is not authorized to recover the secret. If one more player in the complement of $C$ assumes the role of the dealer, the scheme is changes to $L=2$. This operation does not change the fact that $C$ cannot recover the first secret, and thus it is still not authorized for $L=2$. This argument can be continued to any other $1 < L \leq m$ by adding more dealers.
Hence a set of $m-1$ (or fewer) players is not authorized to recover the secret for all value of $1\leq L \leq m$. That a set of more than $m-L$ players is not forbidden follows from the fact that information cannot be lost and thus the complement of a forbidden set has to be authorized. However, we just argued that the complement of a set of more than $m-L$ players is not authorized (since it consists of less than $m$ players). Hence any set with more than $m-L$ and fewer than $m$ players is an intermediate set.
\end{proof}

A closer look at the proof shows us that it actually is not absolutely necessary for the initial state to be maximally entangled with respect to any bipartition, but only for bipartitions for which all dealers are in the same set. In fact, we can generalize the proof of Theorem~\ref{theorem:AME-QSS} to the case of ramp QSS to show that this is a necessary and sufficient condition for the construction of $(m,L,2m-L)$ ramp QSS schemes.

\begin{theorem}
	\label{EntQECC:theorem:AME-ramp}
For a state $\ket{\Phi} \in \HH_P \otimes \HH_R$, shared between $2m-L$ players $P$, each holding a qudit, and $L$ reference qudits, the following two properties are equivalent:
\begin{enumerate}[(i)]
	\item $\ket{\Phi}$ is maximally entangled for any bipartition for which the $L$ reference qudits are in the same set.
	\item $\ket{\Phi}$ is the purification of a $(m,L,2m-L)$ ramp QSS schemes. The encoded secret of the ramp QSS scheme has dimension $d^L$, and each share has dimension $d$.
\end{enumerate}
\end{theorem}

The proof is a straightforward generalization of the proof of Theorem~\ref{theorem:AME-QSS} and is provided in Appendix~\ref{appendix:ramp}.

\section{Open-destination teleportation}
\label{section:openteleport}
Given a state with such high amount of entanglement as the AME state has, one cannot help thinking about ways of using these resources for teleportation protocols. In Ref.~\cite{Helwig2012} we already showed how AME states can be used for two different teleportation scenarios that require either sending or receiving parties to perform joint quantum operations, while the other end may only use local quantum operations.

Another teleportation scenario that uses genuine multipartite entanglement, and has already been demonstrated experimentally \cite{Zhao2004}, is open-destination teleportation. In this scenario, a genuinely multipartite entangled state is shared between $n$ parties, each in the possession of one qudit. One of the parties, the dealer, performs a teleportation operation on her qudit and an ancillary qudit $\ket{\Phi}$. After this teleportation operation, the final destination of $\ket{\Phi}$ is still undecided, thus open-destination teleportation. The destination is decided upon in the next step, where a subset $A$ of the remaining parties $P$ performs a joint quantum operation on their qudits such that a player in $P\backslash A$ ends up with the state $\ket{\Phi}$ -- up to local operations that depend on measurement outcomes of the dealer and parties $A$. Here we want to show that open-destination teleportation can also be performed with AME states.

Assume that an \AME{n,d} state has been distributed among $n$ parties. One of the $n$ parties is assigned the role of the dealer. She performs a Bell measurement on her qudit and the secret $\ket{S} = \sum a_i \ket{i}$. This transforms the state to
\begin{equation}
       \ket{S} \ket{\Phi} \rightarrow
	\ket{\Phi_S} = \frac{1}{\sqrt{d^m}}
	\sum_{(k,i)\in \ZZ_d^m} a_{p q,i} 
	\ket{k_1}_{B_1} \cdots \ket{k_{m-1}}_{B_{m-1}} 
	\ket{\phi(k,i)}_A,
\end{equation}
where $pq$ labels the outcome of the Bell measurement and has to be made public. The remaining $n-1$ parties that share the resulting state have been divided into two sets $A$ and $B$ of size $\ceil{n/2}$ and $m-1 = \floor{n/2}-1$, respectively. Now, after the teleportation operation has been completed, the parties in set $A$ may choose one party $B_i \in B$ as the final destination for the state $\ket{S}$. Then, after performing the joint unitary operation of Equation~\eqref{eq:sortV} followed by a Bell measurement on qudits $A_i$ and $A_m$ with outcome $rs$, the party $B_i$ ends up with the state $\ket{\Phi}_{B_i} = U^\dagger_{rs} U^\dagger_{pq} \ket{S}$, which can be easily transformed to $\ket{S}$ if the measurement results $pq$ and $rs$ are known.

Note that with the parallel teleportation protocol introduced in Ref.~\cite{Helwig2012}, also one of the parties in $A$ can be chosen to receive the state $\ket{S}$. Thus, after the dealer's teleportation operation is completed, any set of size greater or equal $\ceil{n/2}$ can choose any of the remaining $n-1$ parties as the final destination of the teleportation.

\include{swapping}

\section{Conclusion}
In this paper, we have shown the existence of AME states for a wide range of parameters, in particular, the derivation of AME states from classical MDS codes proves that AME states exist for any number of parties if the system dimension is chosen large enough. We have proven an equivalence between AME states and threshold quantum secret sharing schemes. By extending the idea of how threshold QSS schemes follow from the entanglement properties of AME states, we have shown that a wider class, namely ramp QSS schemes can be constructed from AME states. The entanglement requirements to construct ramp QSS schemes are over-satisfied by AME states, and we prove the necessary and sufficient entanglement conditions for the construction of ramp QSS schemes.

Two more applications for AME states have been given in form of open-destination teleportation and entanglement swapping of AME states. The quantum secret sharing and teleportation scenarios that have been presented here and in Ref.~\cite{Helwig2012}, indicate that AME states can be used for a wide variety of quantum information protocols that involve the displacement of quantum states. 

\section*{Acknowledgements}
We acknowledge financial support by NSERC and the CRC program. We also want to thank Hoi-Kwong Lo, Jos\'e Ignacio Latorre, Arnau Riera, and David Gosset for helpful discussions and comments.

\appendix

\section{Entanglement in Ramp QSS Schemes}
\label{appendix:ramp}
Here, we give a generalization of the methods used in Theorem~\ref{theorem:AME-QSS} to prove the equivalence of AME states and threshold QSS schemes to $(m,L,2m-L)$ ramp QSS schemes for arbitrary $L$. The generalization is very straightforward, the secret dimension is now $d^L$ instead of $d$, changing also the dimension of the reference system to $d^L$. We define an isometry $U_S$ that encodes the $d^L$ dimensional secret $S$ into a state shared by the $2m-L$ players, each holding a $d$ dimensional system, 
\begin{equation}
	\label{EntQECC:eq:rampencoding}
	U_S: \HH_S \rightarrow \HH_1 \otimes \dots \otimes \HH_{2m-L},
\end{equation}
where $\HH_i \cong \CC^d$ and $\HH_S \cong {\CC^{d^L}}$.

We further introduce a reference system $\HH_R \cong \HH_S$ and consider the state $\ket{\Phi}$ that is generated by applying the encoding operation to $\HH_S$ for a maximally entangled state $\ket{RS} = 1/\sqrt{d} \sum_i \ket{i}\ket{i} \in \HH_R \otimes \HH_S$, i.e., $\ket{\Phi} = \Id_R \otimes U_S \ket{RS}$. A set of players $A \subset P$ shares the state $\rho_{RA} = \Tr_{P\backslash A} \ket{\Phi}$ with the reference system. $A$ is authorized, if there exists a completely positive map $T_A: \HH_A \rightarrow \HH_S$ such that \cite{Imai2005, Schumacher1996}
\begin{equation}
	\label{EntQECC:eq:rampRA_recover}
	\Id_R \otimes T_A (\rho_{RA}) = \ket{RS}.
\end{equation}
For the mutual information between an authorized set (i.e., $|A| \geq m$) and the reference system is
\begin{equation}
    \label{EntQECC:eq:rampQSSconditionA}
	I(R:A) = I(R:S) = 2S(S) \quad \text{if } |A| \geq m,
\end{equation}
and for a forbidden set, we must have
\begin{equation}
    \label{EntQECC:eq:rampQSSconditionB}
	I(R:B) = 0 \quad \text{if } |B| \leq m-L.
\end{equation}
$U_S$ defines a $(m,L,2m-L)$ ramp QSS scheme if and only if these two equations are satisified.

Since any set of players $C \subset P$ with $|C| = L$ can change some forbidden set into an authorized set, we have  $S(C) \geq S(S)$ \cite{Imai2005} for all sets with $L$ players. And because $S(S)$ is maximal and equal to $S(R)$,
\begin{equation}
	\label{EntQECC:eq:rampsecretentropy}
	S(S) = S(R) = S(C) = L \log d.
\end{equation}
Equations~\eqref{EntQECC:eq:rampQSSconditionA} and \eqref{EntQECC:eq:rampQSSconditionB} can be rewritten to give
\begin{align}
	\label{EntQECC:eq:entropyRA}
	S(R,A) &= S(A) - S(R) \quad \text{if } |A| \geq m\\
	\label{EntQECC:eq:entropyRB}
	S(R,B) &= S(B) + S(R) \quad \text{if } |B| \leq m-L.
\end{align}
This sums up the changes in the lead-up to Theorem~\ref{theorem:AME-QSS}, whose version we may now state and prove for ramp QSS schemes. For this we regard the reference system of dimension $d^L$ as consisting of $L$ systems, each of dimension $d$, so that $\ket{\Phi}$ is a state shared between $2m$ parties, $2m-L$ players that share the secret and $L$ in the reference system, each possessing a qudit.

\newtheorem*{thm:ramp}{Theorem~\ref{EntQECC:theorem:AME-ramp}}
\begin{thm:ramp}
For a state $\ket{\Phi} \in \HH_P \otimes \HH_R$, shared between $2m-L$ players $P$, each holding a qudit, and $L$ reference qudits, the following two properties are equivalent:
\begin{enumerate}[(i)]
	\item $\ket{\Phi}$ is maximally entangled for any bipartition for which the $L$ reference qudits are in the same set.
	\item $\ket{\Phi}$ is the purification of a $(m,L,2m-L)$ ramp QSS schemes. The encoded secret of the ramp QSS scheme has dimension $d^L$, and each share has dimension $d$.
\end{enumerate}
\end{thm:ramp}
\begin{proof}
$(i) \rightarrow (ii)$:
In the equations for the mutual information, all occurring sets, $A$, $B$, $R$, $A \cup R$ and $B \cup R$, are maximally entangled with the rest because for all of them all reference qudits are in the same set of the bipartition. Hence we have $S(A) = (2m-|A|) \log d$, $S(B) = |B| \log d$, $S(R) = S(S) = L \log d$, $S(A,R) = (2m-|A|-L) \log d$ and $S(A,B) = (|B| + L) \log d$. Plugging these into Equations~\eqref{EntQECC:eq:rampQSSconditionA} and \eqref{EntQECC:eq:rampQSSconditionB} while using the definition of the mutual information (Equation~\ref{eq:mutualinfo}), confirms that these are satisfied.

$(ii)\rightarrow (i)$:
Consider an unauthorized set of players $B$, with $|B| = m-L$. Then the set is $B \cup C$ is authorized for any additional set $C$ with $|C| = L$ and $C \cap B = \emptyset$. From Equation~\eqref{EntQECC:eq:entropyRA} we have
\begin{equation}
 S(B,C,R) = S(B,C) - S(R)
\end{equation}
On the other hand, using the Araki-Lieb inequality \cite{Nielsen2000} $S(X,Y) \geq S(X) - S(Y)$ and Equation~\eqref{EntQECC:eq:entropyRB} gives
\begin{equation}
 S(B,C,R) \geq S(B,R) - S(C) = S(B) + S(R) - S(C).
\end{equation}
Combining the last two equations and using $S(S)=S(R)=S(C)$ shows
\begin{equation}
 S(B,C) \geq S(B) + S(C),
\end{equation}
where equality must hold due to the subadditivity of the entropy $S(X,Y) \leq S(X) + S(Y)$. This means that the entropy increases maximally when adding $L$ shares to $m-L$ shares. The strong subadditivity of the entropy \cite{Nielsen2000}
\begin{equation}
 S(X,Y) - S(Y) \geq S(X,Y,Z) - S(Y,Z)
\end{equation}
states that adding system $X$ to system $Y$ increases the entropy at least by as much as adding system $X$ to a larger system $Y \cup Z$ that contains $Y$. So in our case, adding $L$ shares to less than $m-L$ shares increases the entropy by at least $S(C)$, and since this is the maximum, it increases the entropy exactly by $S(C)$. Moving the shares over one by one from $C$ to $m-L$ or less shares must increase the entropy maximally with each share for it to be maximally increased when all shares are added. Hence adding one share to a set that contains less than $m$ shares increases the entropy maximally. Hence, starting out with a set of no shares, and repeatedly adding one share to the set until the set contains any $m$ shares and is authorized, shows that any set of $m$ shares has entropy $m \log d$. This shows that the entropy is maximal for any subset of $m$ players, i.e., $\ket{\Phi}$ is maximally entangled for any bipartition into $m$ players $A$ and its complement $P\backslash A \cup R$, which 
contains all $L$ reference qudits, and thus is maximally entangled for any bipartition where all reference qudits are in the same set.
\end{proof}

\bibliographystyle{ieeetr}
\bibliography{entanglement}

\end{document}

%% file: swapping.tex
\section{Swapping of AME states}
\label{section:swapping}

Entanglement swapping \cite{Yurke1992} is a very useful tool for the application of entanglement in communication. By making a Bell measurement on Bob's side, two entangled states shared between Alice and Bob, and Bob and Charlie, respectively, can be transformed into an entangled state shared by Alice and Charlie. Employing this procedure in quantum repeaters \cite{Dur1999} allows entangled states to be used for long distance communications. In this section, we show to what extent a generalization of the entanglement swapping protocol can be constructed to allow swapping of entanglement between absolutely maximally entangled states shared between different parties.

Assume that parties $\{1, 2, \ldots, 2n\}$ share an \AME{2n,d} state,
\begin{eqnarray}
\ket{\Phi}_{1,\ldots, 2n}&=&\sum \ket{i_1\cdots i_n}_{1,\ldots, n}\ket{\phi(i_1, \ldots, i_n)}_{n+1,\ldots, 2n}\\
                            &=&\sum \ket{i_1 \cdots i_n}_{1,\ldots, n}U\ket{i_1\cdots i_n}_{n+1, \ldots, 2n},
\end{eqnarray}
where $U$ is a unitary transformation with $U\ket{i_1\cdots i_n}=\ket{\phi(i_1, \ldots, i_n)}$. 

Suppose parties $\{n+1, \ldots, 3n\}$ also share an \AME{2n,d} state
\begin{eqnarray}
\ket{\Phi}_{n+1,\ldots, 3n}
                            &=&\sum \ket{i_1 \cdots i_n}_{n+1, \ldots, 2n}U\ket{i_1\cdots i_n}_{2n+1,\ldots, 3n}.
\end{eqnarray}
Now each of the parties $\{n+1, \ldots, 2n\}$ performs a Bell measurement on their qudits from both AME states. Without loss of generality, we can assume the measurement result is $(q,p)=(0,0)$ (see Equation~\eqref{eq:Bell} for the notation), since other measurement outcomes produce the same state up to local transformations. Then the state shared by the parties $\{1, \ldots, n, 2n+1, \ldots, 3n\}$ becomes
\begin{equation}
\ket{\Phi}_{1, \ldots, n, 2n+1, \ldots, 3n}=\sum \ket{i_1\cdots i_n}_{1,\ldots, n}U^2\ket{i_1\cdots i_n}_{2n+1, \ldots, 3n}
\end{equation} 
Consecutive applications of the above procedure gives the following lemma:
\begin{lemma}
\label{lemma:swapping}
Suppose each group of parties $\{1, \ldots, 2n\}$, $\{n+1, \ldots, 3n\}$, $\cdots$, $\{mn+1, \ldots, (m+1)n\}$ shares an \AME{2n,d} state, 
\begin{equation}
\ket{\Phi}=\sum \ket{i_1 \cdots i_n}U\ket{i_1\cdots i_n}.
\end{equation}
Then, if each of the parties $\{n+1, n+2, \ldots, mn\}$ performs a Bell measurement on their two qudits, the resulting state shared by the parties $\{1,\ldots, n, mn+1, \ldots, (m+1)n\}$ is locally equivalent to 
\begin{equation}
\label{eq:swappinglemma}
\ket{\Phi}_{1,\ldots, n, mn+1, \ldots, (m+1)n}=\sum \ket{i_1\cdots i_n}_{1, \ldots, n}U^m\ket{i_1\cdots i_n}_{mn+1, \ldots, (m+1)n}
\end{equation}
\end{lemma}
\begin{proof}[Proof by induction]
The case for $m=2$ is demonstrated in the above discussion already. If the lemma holds for $m$, for $m+1$ the two remaining states, after the parties $\{n+1, n+2, \ldots, mn\}$ performed their Bell measurements, are 
\begin{multline}
\ket{\Phi}_{1,\ldots, n, mn+1, \ldots, (m+1)n}=\\
\sum \ket{i_1\cdots i_n}_{1, \ldots, n}U^m\ket{i_1\cdots i_n}_{mn+1, \ldots, (m+1)n}
\end{multline} 
and
\begin{multline}
\ket{\Phi}_{mn+1, \ldots, (m+1)n,(m+1)n+1, \ldots, (m+2)n}=\\
\sum \ket{i_1\cdots i_n}_{mn+1, \ldots, (m+1)n}U\ket{i_1\cdots i_n}_{(m+1)n+1, \ldots, (m+2)n}.
\end{multline} 
After the parties $\{mn+1, \ldots, (m+1)n\}$ all perform a Bell measurement, the state shared by $\{1, \ldots, n, (m+1)n+1, \ldots, (m+2)n\}$ becomes
\begin{multline}
\ket{\Phi}_{1,\ldots, n, (m+1)n+1, \ldots, (m+2)n}=\\
\sum \ket{i_1\cdots i_n}_{1, \ldots, n}U^{m+1}\ket{i_1\cdots i_n}_{(m+1)n+1, \ldots, (m+2)n}.
\end{multline}
\end{proof}
The state in Equation~\eqref{eq:swappinglemma} is generally not an AME state, however, depending on the exact form of the unitary $U$, the resulting state can be absolutely maximally entangled again for certain $m$, as expressed in the following corollary.
\begin{corollary}[Swapping of AME States]
\label{corollary:swapping}
Suppose each set of parties $\{1, \ldots, 2n\}$, $\{n+1, \ldots, 3n\}$, $\cdots$, $\{mn+1, \ldots, (m+1)n\}$ shares an \AME{2n,d} state, 
\begin{equation}
\ket{\Phi}=\sum \ket{i_1 \cdots i_n}U\ket{i_1\cdots i_n}.
\end{equation}
If $U^m$ is locally unitary equivalent to $U$ up to some permutation of parties, then, by making a Bell measurement on each of the parties $(n+1, \ldots, mn)$, parties $(1, \ldots, n, mn+1, \ldots, (m+1)n)$ will share an \AME{2n,d} state.
\end{corollary}

In the following we will show an example for AME swapping, the swapping of an \AME{4,3} state. As an application we will show that different from the EPR state, the \AME{4,3} state 
\begin{equation}
\begin{split}
\label{eq:swapping43state}
\ket{\Phi} &= \ket{0000}+ \ket{0111}+ \ket{0222}\\
                         &+ \ket{1012}+ \ket{1120} + \ket{1201}\\
                         &+ \ket{2021}+ \ket{2102} + \ket{2210}
\end{split}
\end{equation}
needs two steps of Bell measurements for the swapping to reproduce an \AME{4,3} state.

\begin{example}
Assume we have three \AME{4,3} states, shared by the players $\{A, B, C, D\}$, $\{C, D, E, F\}$, and $\{E, F, G, H\}$, respectively. After $C$, $D$, $E$, and $F$ all perform a Bell measurement on their two qutrits, the parties $\{A, B, G, H\}$ will share an \AME{4,3} state.  This is illustrated in Figure~\ref{fig:swapping}.
\end{example}

\begin{figure}
	\centering
	\includegraphics[width=.5\textwidth]{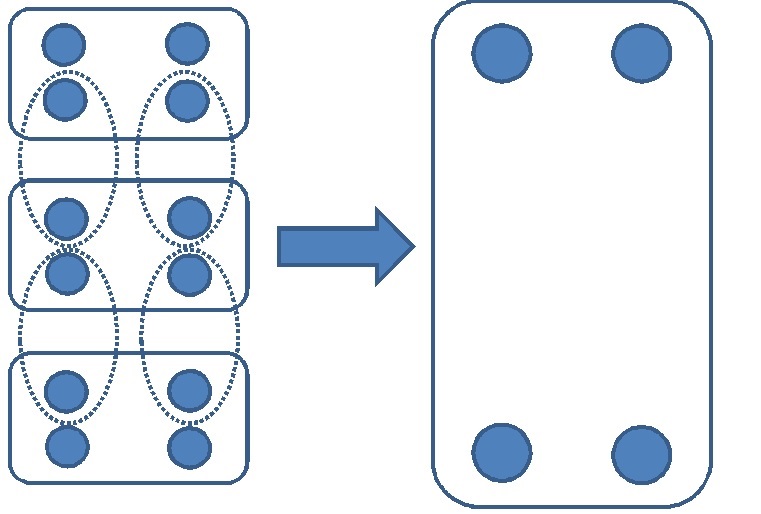}
	\caption{Entanglement swapping between three \AME{4,3} states results in a new \AME{4,3} state between previously unentangled parties. Dotted circles indicate where Bell measurements have to be performed.}
	\label{fig:swapping}
\end{figure}

\begin{proof}[Calculation]
From Equation~\eqref{eq:swapping43state} we can see that the unitary $U$ of Lemma~\ref{lemma:swapping} is given by
\begin{equation}
\begin{split}
U\ket{00}=\ket{00} \qquad U\ket{01}=\ket{11} \qquad U\ket{02}=\ket{22}\\
U\ket{10}=\ket{12} \qquad U\ket{11}=\ket{20} \qquad U\ket{12}=\ket{01}\\
U\ket{20}=\ket{21} \qquad U\ket{21}=\ket{02} \qquad U\ket{22}=\ket{10}
\end{split}
\end{equation}
Applying this unitary twice gives
\begin{equation}
\begin{split}
U^2\ket{00}=\ket{00} \qquad U^2\ket{01}=\ket{20} \qquad U^2\ket{02}=\ket{10}\\
U^2\ket{10}=\ket{01} \qquad U^2\ket{11}=\ket{21} \qquad U^2\ket{12}=\ket{01}\\
U^2\ket{20}=\ket{21} \qquad U^2\ket{21}=\ket{02} \qquad U^2\ket{22}=\ket{10}
\end{split}
\end{equation}
It can be easily seen that by permuting the two parties and exerting a unitary transformation that exchange $\ket{1}$ and $\ket{2}$ in the second party, this unitary transformation becomes the identity. Thus $U^3$ is locally unitary equivalent to $U$ up to permutation of parties, and together with Corollary~\ref{corollary:swapping}, it is easily to see that after the Bell measurement, the resulting state would be locally unitary equivalent with an \AME{4,3} state up to a permutation of party $G$ and $H$. Since the \AME{4,3} state satisfies permutation symmetry, which means by permuting any two parties the resulting state is still an \AME{4,3} state, $A, B, G, H$ really share the same \AME{4,3} state that was swapped. 
\noqed
\end{proof}

\begin{remark}
In the above example, we require that party $C$ acts as the third party of the first AME state and the first party of the second AME state. This is actually not required. Since the \AME{4,3} state is permutationally invariant, we only need $C$ to posses any qudit of each AME state. The same requirement applies for $D$. In fact, most of the AME states we found are permutational invariant, and in these cases we do not need to have restriction on which specific qudits the parties control. 
\end{remark}